\documentclass[letterpaper,aps,superscriptaddress,10pt,tightenlines,longbibliography,onecolumn,notitlepage,accepted=2021-01-12]{quantumarticle}
\pdfoutput=1
\usepackage{amsmath,amsthm,amsfonts,amssymb,braket}
\usepackage{graphicx}

\usepackage[usenames,dvipsnames]{color}
\usepackage[colorlinks=true,citecolor=blue,linkcolor=blue]{hyperref}
\usepackage[capitalize,nameinlink]{cleveref}

\newtheorem{lemma}{Lemma}

\usepackage{algorithm}

\newcommand{\la}{\langle}
\newcommand{\ra}{\rangle}
\newcommand{\Sw}{{\rm Swap}}

\newcommand{\cO}{{\cal O}}

\newcommand{\expec}{\mathbb{E}}

\newcommand{\mU}{{\mathcal{U}}}
\newcommand{\mM}{{\mathcal{M}}}
\newcommand{\rd}{{\mathrm d}}
\newcommand{\RR}{{\mathbb R}}

\DeclareMathOperator{\tr}{{\mathrm tr}}
\DeclareMathOperator{\Tr}{{\mathrm Tr}}

\begin{document}

\title{How Dynamical Quantum Memories Forget}

\author{Lukasz Fidkowski}
\affiliation{Department of Physics, University of Washington, Seattle, WA 98195-1560, USA}

\author{Jeongwan Haah}
\affiliation{Microsoft Quantum and Microsoft Research, Redmond, WA 98052, USA}

\author{Matthew B.~Hastings}
\affiliation{Station Q, Microsoft Research, Santa Barbara, CA 93106-6105, USA}
\affiliation{Microsoft Quantum and Microsoft Research, Redmond, WA 98052, USA}

\begin{abstract}
Motivated by recent work showing that 
a quantum error correcting code can be generated by hybrid dynamics of unitaries and measurements, 
we study the long time behavior of such systems. 
We demonstrate that even in the ``mixed'' phase, 
a maximally mixed initial density matrix is purified on a time scale equal to the Hilbert space dimension 
(i.e., exponential in system size), 
albeit with noisy dynamics at intermediate times
which we connect to Dyson Brownian motion.
In contrast, we show
that free fermion systems --- 
i.e., ones where the unitaries are generated by quadratic Hamiltonians 
and the measurements are of fermion bilinears --- 
purify in a time quadratic in the system size.  In particular, a volume law phase for the entanglement entropy cannot be sustained in a free fermion system.
\end{abstract}

\maketitle

Recently it has been argued that a low-dimensional (even a one-dimensional) quantum system which mixes unitary evolution by local circuits with local measurements can act as a quantum memory~\cite{Fisher1,Ruhman,Fisher, Smith, Gullans, Altman, Fan}.  If one records the outcomes of the measurements, this process can protect nontrivial quantum information.  Here, we investigate the long-time dynamics of this process to understand how the system ultimately ``forgets,'' {\emph{i.e.,}} if the system is used to store quantum information, how the information necessarily is lost by these measurements.

To study this long time dynamics, we ignore the spatial structure.  The system consists of just a single Hilbert space of high dimension $N$, with $N$ even.  Our model consists of alternating two different steps: first, a unitary evolution, followed by a measurement of a single bit of information\footnote{Of course, one might generalize this to a two-outcome POVM.  Since such a POVM is equivalent to a projective measurement in a larger Hilbert space, the POVM model is equivalent to our model, except for some change in the ensemble from which we choose the unitary evolution.  We leave this for future work.}, represented by a rank $N/2$ projector.  We can also choose to conjugate the measurements by the unitary, and so the model can be described by measuring a single bit of information at each step, with the measurement basis changing each time.  Thus, if we evolve by unitary $U_1$, then measure projector $P_1$, then evolve by unitary $U_2$, then measure projector $P_2$, this is equivalent, up to an overall unitary, to measuring projector $U_1^\dagger P_1 U_1$, followed by measuring projector $U_1^\dagger U_2^\dagger P_2 U_2 U_1$.  We keep track of the quantum trajectories by writing down the measurement outcomes, so in particular pure states always evolve to pure states along such trajectories.  

We consider two different cases, that we term ``many-body" and ``free fermion".  In the many-body case, the unitaries as chosen to be Haar random.  The term ``many-body" is a bit of a misnomer: we have some fixed high-dimensional Hilbert space, perhaps formed by tensoring many qubits, so a better term might be ``high-dimensional single body".  Nevertheless, we persist in using the term many-body; in particular, one may hope that sufficiently deep quantum circuits for a tensor product Hilbert space can be well-approximated by our Haar random measurements~\cite{brandao2016local,HarrowMehraban2018,Haferkamp2020}.
In the free fermion case, the Hilbert space is a Fock space of fermions, and measurements are only allowed to be of fermion bilinears.   

In the many-body case we will find that the system preserves information up to a time scale proportional to $N$ (which, for a tensor product Hilbert space, is exponentially large in the number of degrees of freedom).  However, up to that time, we will find large sample-to-sample fluctuations in how well information is preserved.
In contrast, in the free fermion case, on a system with $n$ modes (hence, $2^n$-dimensional Hilbert space), we prove that the purification time is $\sim n^2$.  While not exponential, this time is still slower by a factor of $n$ than the purification time with optimally chosen measurements, which is only proportional to $n$.  Indeed, if measurements can be done in parallel, one can purify a many-body or free fermion system in a single step with $n$ commuting measurements, such as measuring the Pauli $Z$ operator on each of the $n$ qubits.

\begin{figure}
\includegraphics[width=4in]{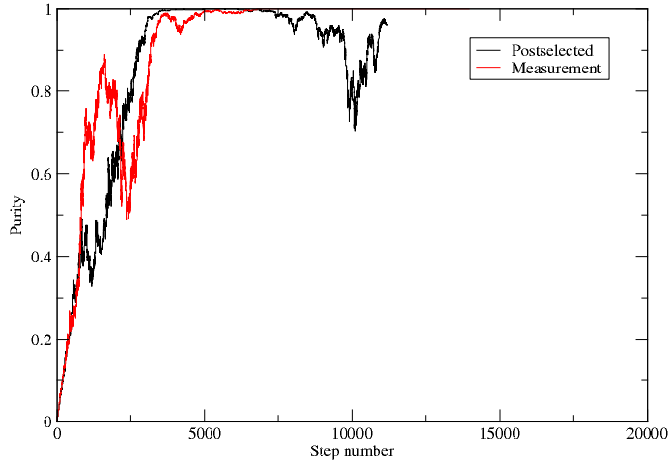}
\caption{Purity as a function of number of steps for $N=2000$ dimensional Hilbert space, starting with maximally mixed density matrix.  Curves show postselected and measurement cases for one run.  At late time, the red measurement curve converges exponentially to $1$ and is no longer visible on the plot.}
\label{fignum}
\end{figure}

To understand how well the system preserves quantum information, we will use a reference system.  We will start with a state that is maximally entangled between the system and the reference, so that the reduced density matrix $\rho$ on the system is maximally mixed.  We then study the evolution of this reduced density matrix under a sequence of measurements~\cite{Gullans}.  Before describing results, it is helpful to see the output of a numerical simulation of the many-body case for a Hilbert space of dimension $N=2000$, shown in~\cref{fignum}.  At each step we choose a rank $N/2$ projector randomly, by choosing $P=U P_0 U^\dagger$ for some fixed projector $P_0$ with $U$ a unitary chosen from Haar measure.  The red ``measurement" curve is the physically meaningful case: we compute the probability $\tr(P_0 \rho)$ that the measurement has a given outcome, and with that probability replace the density matrix $\rho$ with $(P_0 \rho P_0) / \tr(P_0 \rho) $, otherwise replacing it with $(1-P_0) \rho (1-P_0) / \tr((1-P_0)\rho)$.  The black ``post-selection" curve is not physically meaningful; here we always replace $\rho$ with $P_0 \rho P_0 / \tr(P_0 \rho)$.  The quantity being plotted is the ``purity" of $\rho$, defined to be $\tr(\rho^2)$.

One sees that the dynamics has roughly three regimes.
First, at early times in both cases the purity grows linearly with only small fluctuations.
Then, at intermediate times, the purity has noisy dynamics, and can actually {\it decrease}.  Indeed, one can observe long fluctuations in which the purity decreases for many steps.  Finally, the purity becomes close to $1$ and converges to $1$.
The long time convergence is quite different between the measurement and postselected case, with an exponential convergence for measurement, but not for postselection.  This will follow from \cref{pconv,mconv}.

Perhaps the most interesting question is how the decrease in purity can occur.  For one thing, it seems strange as it suggests that a random measurement can actually restore quantum information that has been lost.  Indeed, it is possible for a measurement to reduce purity (or increase entropy), but importantly, the square-root purity {\it averaged over measurement outcomes} cannot decrease after a  measurement (and similarly, the entropy averaged over measurement cannot increase).  We show this in~\cref{entineq}.

Still, though, it may be surprising that large fluctuations occur for random choices of projector, and that they are still present even for $N=2000$.  One heuristic explanation as to why they are present is that as the purity increases, the system starts to get a few larger eigenvalues, and this reduces the tendency of the system to self-average.  Our analysis in~\cref{secMBD} shows mathematically what happens: the average purity increases in a single step by an amount proportional to order $1/N$.  There are terms which are of order $1/N^2$ and smaller, but we ignore those.  At the same time, the variance of the purity after a single step is of order $1/N$ also, i.e. the root-mean square is $1/\sqrt{N}$.
Thus, it suggests the picture that the purity obeys a biased diffusion equation; the bias and the noise both depend on purity.  At a time scale of order $N$, both bias and noise are equally important; that is, the Peclet number is of order $1$, independent of $N$ at this time scale.
We show that for a maximally mixed state, as well as a nearly pure state, the variance is negligible, which explains why the early and late time dynamics are approximately noiseless.

This picture that the purity obeys a biased diffusion equation is not completely correct: in general, the drift and diffusion of the purity depend on traces of higher powers so one cannot write down a Markovian evolution for the purity alone.
In the special case that the density matrix has rank $2$, one can describe the dynamics just in terms of the purity: see \cref{figrank2} and \cref{SD}.
For low rank density matrices, we are able to describe the dynamics of the set of eigenvalues by a diffusion equation with drift.

\begin{figure}
\centering
\includegraphics[width=4in]{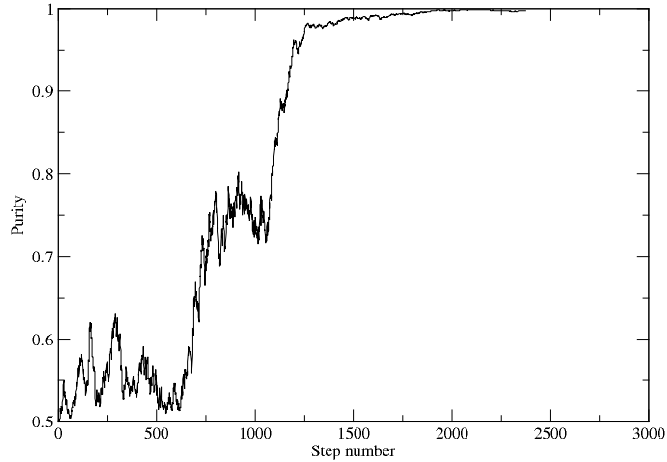}
\caption{Purity as a function of number of steps for $N=2000$ dimensional Hilbert space starting with rank $2$ density matrix with purity $1/2$ and performing random measurements.}
\label{figrank2}
\end{figure}

In the case of free fermion dynamics, we consider two cases: (1) particle-number conserving dynamics and (2) general free fermion evolution, which may involve pairing.  In both cases we use the formalism of Gaussian states \cite{Bravyi}.  The free fermion unitaries are chosen from $U(N)$ or $SO(2N)$ respectively, randomly with respect to Haar measure.  We use a proxy for the second Renyi entropy which is only well defined for Gaussian states but has the advantage of being much easier to work with.  This proxy entropy is equal to the second Renyi entropy for the maximally mixed state and all pure Gaussian states, and stays within fixed positive bounds of it for all Gaussian states, so understanding this proxy entropy is just as good as understanding the actual second Renyi entropy (and hence purity) insofar as the purification dynamics goes.  Our main result is that if we start with a state on $n$ modes whose proxy entropy is $s\, n \log 2$ --- i.e., the entropy density is $s \log 2$ --- then measuring a single mode must decrease the proxy entropy by at least an amount of order $s^2$, when averaged over the two possible measurement outcomes.  This leads to a rigorous bound on the purification time of order $n^2$, showing that this maximally entangling free fermion system is in a purifying phase ~\cite{Gullans, Gullans2}.  We expect, but do not prove, that any free fermion system, no matter how non-local, will purify in a time polynomial with $n$.  Our result is also consistent with the lack of a volume law phase in free fermion unitary-measurement dynamics \cite{DeLuca, Lucas, Khemani, Skinner}.

The paper is organized as follows.  In \cref{secMBD} we give results for average change in purity and fluctuations in purity for the many-body case, and analyze these results.  
In \cref{sec:Dyson} we give a Brownian motion picture for eigenvalues of the density matrix, valid when the rank of the matrix is small compared to $N$.
In \cref{secFFD} we discuss the free-fermion case.  Finally, we give a brief discussion in \cref{discuss}.  In \cref{entineq}, we prove that entropy, averaged over measurement outcomes, decreases after measurement.  In \cref{entineq} we prove general entanglement inequalities.  In \cref{SD} we derive necessary formulas for averages of products of traces over choices of unitary matrices, using a set of ``Schwinger-Dyson equations"\cite{hastings2007random}, for which we give a self-contained derivation.  There are of course many tools one could use, such as the Weingarten calculus, and other readers may prefer that.  The Schwinger-Dyson equations however have the advantage of being relatively simple to use and of naturally organizing the result in powers of $N^{-1}$.

{\it Note:} After finishing this preprint, we became aware of Ref. \cite{LiFisher} which found a similar exponential time purifying behavior using a capillary-wave description of domain walls in an effective statistical mechanics model describing the mixed phase of a 1d local hybrid measurement-unitary circuit.

\section{Many-Body Dynamics}
\label{secMBD}

\subsection{Post-selected case}
We wish to compute the entropy of
$$\rho'=\tr(P \rho)^{-1} P\rho P.$$ where $\rho$ is a density matrix on an $N$ dimensional Hilbert space, and $P=UP_0 U^\dagger$ is a random rank $N/2$ projector.  Here $P_0$ is a fixed rank $N/2$ projector and $U$ a Haar-random unitary.  For computational purposes, we use the `purity', {\emph{i.e.}} the trace of the square of the density matrix, as our entropy measure.  The purity is maximal for a pure state, where it is $1$, and achieves its minimum of $1/N$ for a maximally mixed state.

At large $N$ we expect $\tr P\rho $ to be close to $1/2$, so we write $\tr P \rho =1/2+\delta$.  Then
\begin{align}
\label{series}
\expec \left[\tr({\rho'}^2)\right] =\expec \left[\frac{\tr P\rho P \rho}{(\tr P\rho)^2}\right]=\expec[4\tr P\rho P\rho-16(\tr P \rho P \rho) \delta +48(\tr P \rho P \rho) \delta^2 +O(\delta^3)].
\end{align}
where $\expec$ denotes the average over the unitaries $U$.  Let us estimate how big $\delta$ can be.  Using \cref{newap2} of \cref{SD} we have that

\begin{align}\label{delsquared}
\expec[\delta^2]=\expec \left[(\tr P \rho)(\tr P \rho)-\frac{1}{4}\right]=\frac{1}{N}\left(\frac{1}{2} \tr \rho^2 - \expec[\tr P\rho P\rho]\right),
\end{align}
so $\expec[\delta^2]$ is $O(N^{-1})$ and thus with high probability $\delta$ is smaller than $N^{-1/2+\epsilon}$ for any $\epsilon>0$.

In fact, we can give even better bounds on the fluctuations in $\delta$ using concentration of measure.  We have $\expec[\delta]=0$.  Regard $\delta=\tr P\rho -1/2=\tr U P U^\dagger \rho-1/2$ as a function of $U$. By a triangle inequality, this function is $2$-Lipschitz using the operator norm as a metric for $U$\footnote{Proof: we wish to bound $|\tr (U P U^\dagger \rho)-\tr(V P V^\dagger \rho)|$
for unitary $U,V$ given a bound on the operator norm $\Vert U-V\Vert$.  
We will prove the bound in a more general case allowing $U,V$ to be non-unitary but requiring that $\Vert U \Vert,\Vert V \Vert \leq 1$.
We have $|\tr(U P U^\dagger \rho)-\tr(V P V^\dagger \rho)| \leq |\tr((U-V) P U^\dagger \rho)| + |\tr(U P (U-V)^\dagger \rho)|+|\tr((U-V) P (U-V)^\dagger \rho|$.  Consider the first term.  We have $\Vert (U-V) P U^\dagger \Vert \leq \Vert U-V \Vert$ so the first term is bounded by $\Vert U-V \Vert$.  The second term is bounded similarly.  The third term is bound by $\Vert U-V\Vert^2$.  So, 
$|\tr(U P U^\dagger \rho)-\tr(V P V^\dagger \rho)|\leq 2 \Vert U-V\Vert+\Vert U-V\Vert^2$.  We claim that we can ignore the last $\Vert U-V \Vert^2$ term; indeed, given any $U,V$, and any integer $k$ consider the sequence of operators $U_0=U,U_1,U_2,\ldots,U_k=V$ given by linearly interpolating between $U$ and $V$ so that $\Vert U_i-U_{i+1}\Vert=\Vert U-V\Vert/k$.  In the limit $k\rightarrow \infty$, the $\Vert U-V\Vert^2$ term can be dropped.}
Hence, this function
is also $2$-Lipschitz using Hilbert-Schmidt norm as a metric.  Hence, (see for example theorem 5.17 of \cite{meckes2019random}), the probability that $|\delta|\geq x$ for any $x>0$ is bounded by $\exp(-\Omega(N x^2))$.  So, the terms of order $\delta^3$ are bounded by $O({\rm poly}(\log(N))/N^{3/2})$ with probability $1-1/N^{\alpha}$ for any constant $\alpha$ and so we may neglect them when estimating to order $1/N$.  We formalize this into the following `non-perturbative' result:
\begin{lemma} \label{lemma1}
For any $\alpha>1$, the following holds:
\begin{align} \label{betterseries}
\expec \left[\frac{\tr P \rho P \rho}{(\tr P \rho)^2}\right] &=\expec[4\tr P\rho P\rho -16(\tr P \rho P \rho) \delta +48(\tr P \rho P \rho) \delta^2] \\
  &+ \rm{max}\,\left( O({\rm poly}(\log(N))/N^{3/2}), O(N^{-\alpha})\right). \nonumber
\end{align}
\begin{proof}
We bound the expectation value of the ``error":
$$\Bigl| \frac{\tr P \rho P \rho}{(\tr P\rho)^2}-\Bigl(4\tr P\rho P\rho -16(\tr P \rho P \rho) \delta +48(\tr P \rho P \rho) \delta^2\Bigr)\Bigr|$$
by dividing into two cases. 

{\bf (1)} $\delta$ is sufficiently small, namely $O({\rm poly}(\log(N))/N^{1/2})$, which happens with probability  $\geq 1-1/N^{\alpha}$.  
In this case, the error is bounded by $\Bigl| \frac{\tr P \rho P \rho}{(\tr P\rho)^2}-(4\tr P\rho P\rho -16(\tr P \rho P \rho) \delta +48(\tr P \rho P \rho) \delta^2)\Bigr| \leq {\rm poly}(\log(N))/N^{3/2}$.

{\bf (2)} $\delta$ is not sufficiently small.  This happens with probability $1/N^{\alpha}$.  
However, for {\it any} choice of $P$ with $\tr P\rho \neq 0$, the error 
$\Bigl| \frac{\tr P \rho P \rho }{(\tr P \rho)^2}-(4\tr P\rho P\rho- 16(\tr P \rho P \rho) \delta +48(\tr P \rho P \rho) \delta^2)\Bigr|$
is bounded by a constant (using a triangle inequality, it is trivially bounded by $1+4+16+48$).
Hence the contribution to the expectation value of the error is bounded by $O(N^{-\alpha})$.

\end{proof}
\end{lemma}
In \cref{derivedsum} \cref{SD} we compute the first three terms in \cref{betterseries} to order $1/N$, using the Schwinger-Dyson equations to perform averages over $U$, to obtain:

\begin{align} \label{postsel}
\expec \left[\frac{\tr P \rho P \rho }{(\tr P \rho)^2}\right]= \tr \rho^2 +\frac{1}{N} \left(1-4 \tr \rho^3+3 (\tr \rho^2)^2\right)+\ldots
\end{align}
where the dots represent terms the asymptotically small error term in \cref{betterseries}.  When $\rho$ is close to a maximally mixed state, the traces of the higher powers of $\rho$ on the right hand side of \cref{postsel} can be neglected, and we see that the purity of $\rho$ increases by $1/N$.  This explains the initial linear growth of purity in \cref{fignum}.

Note that by monotonicity of Schatten $p$-norms, $(\tr \rho^2)^{1/2} \ge (\tr \rho^3)^{1/3}$, so that in the large $N$ limit
\begin{align}
\expec \left[\frac{\tr P \rho P \rho }{(\tr P \rho)^2}\right] \geq \tr \rho^2 +\frac{1}{N}(1-4(\tr \rho^2)^{3/2}+3(\tr \rho^2))^2). \nonumber
\end{align}
When $\rho$ is close to a pure state, {\emph{i.e.}} for $\tr \rho^2=1-\epsilon$ with $\epsilon$ small, this gives
\begin{align} \label{pconv}
\expec \left[\frac{\tr P \rho P \rho}{(\tr P \rho)^2}\right]  \geq 1- \left(\epsilon - \frac{3}{2N} \epsilon^2\right)+\cdots,
\end{align}
where now the $\cdots$ denote higher order terms in $\epsilon$ and $N^{-1}$.  Naively iterating this would give $\epsilon$ evolving in time $t$ (number of steps) as $N/t$, implying purification at a time scale scale $t \sim N^2$.  However, we cannot ignore non-linearities in the noise at those time scales, so we cannot draw any sharp conclusions about the late time behavior in the post-selected case.

In \cref{apnoise1} of \cref{SD} we also calculate the noise to leading order in $1/N$:

\begin{align}\label{noise1}
\expec \left[\frac{(\tr P\rho P\rho)^2}{(\tr P\rho)^4}\right]-\expec \left[\frac{(\tr P\rho P\rho)}{(\tr P \rho)^2}\right]^2 &= \frac{4}{N} \left( \tr \rho^4-2(\tr \rho^3)(\tr \rho^2) + (\tr \rho^2)^3\right) + \ldots
\end{align}
Again, the dots indicate an error term that is bounded to be asymptotically smaller than the leading $1/N$ term by \cref{lemma1}.  When $\tr \rho^2$ is small, $\tr \rho^3$ and $\tr \rho^4$ are upper bounded by $(\tr \rho^2)^{3/2}$ and $(\tr \rho^2)^2$ respectively, and lower bounded by $0$, which shows that, to leading order in $\tr \rho^2$, the noise is upper bounded by $4 (\tr \rho^2)^2 / N$.  This explains the lack of noise during the initial growth of the purity in \cref{fignum}.  Conversely, for a nearly pure state, $\tr \rho^2=1-\epsilon$, we show in \cref{noisebound} in \cref{SD} that the noise is upper bounded by $4 \epsilon^2/N$.  This is consistent with the lack of noise when the purity is close to $1$ in \cref{fignum}.

\subsection{Measurement case}

Now let us consider the case without post-selection.  We want to compute

\begin{align}
\expec \left[(\tr P\rho) \frac{\tr P \rho P \rho }{(\tr P\rho)^2}+(\tr (I-P)\rho) \frac{\tr (I-P) \rho (I-P) \rho}{(\tr (I-P)\rho)^2}\right],
\end{align}
 where $I$ is the identity matrix.  This is the desired value since the measurements outcomes occur with probability $\tr P \rho$ and $\tr (I-P)\rho$ respectively.  After cancelling the probability against one power of the trace in the denominator, and using the fact that the probability distribution for $P$ is invariant under $P\rightarrow I-P$, we want
 
 \begin{align}
\expec \left[2 \frac{\tr P\rho P\rho}{\tr P\rho}\right] &=\expec \left[4\tr P\rho P\rho-8(\tr P \rho P \rho) \delta +16(\tr P \rho P \rho) \delta^2 +\ldots \right]
\end{align}
where again the dots represent an asymptotically small error bounded in the same way as that in \cref{lemma1}.  As shown in \cref{derivedmeas} in \cref{SD}, the result is

\begin{align} \label{meas}
\expec \left[2 \frac{\tr P\rho P\rho}{\tr P\rho}\right] &=\tr \rho^2+\frac{1}{N}\left[1-2(\tr \rho^3)+(\tr \rho^2)^2\right] + \ldots
\end{align}
Once again we see a linear initial growth in steps of $1/N$, 
consistent with the red measurement curve in \cref{fignum}.
The second term is in fact larger than $\frac 1 N [ 1 - \tr \rho^2]$;
this is because 
$1 - 2 (\tr \rho^3) + (\tr \rho^2)^2 
\ge 
1 - 2(\tr \rho^2)^{3/2} + (\tr \rho^2)^2 
= 1 - (\tr \rho^2) + (\sqrt{\tr \rho^2} - \tr \rho^2)^2
$
by the monotonicity of the Schatten $p$-norms.
Thus, 
setting 
\[
F = 1 - \tr \rho^2,
\]
we see that the change $\Delta F$ of $F$ in the large $N$ limit obeys
\begin{align} \label{mconv}
\expec \left[\Delta F \right]\le - \frac{1}{N} F + \cdots .
\end{align}
This implies that for
\emph{any} initial probabilistic ensemble of $\rho$,
the ensemble average purity converges to~$1$ exponentially
with characteristic time $N$.
This has two immediate consequences.

First, the impurity $F$ is exponentially small with probability exponentially close to $1$.
Indeed, since $F \geq 0$, Markov's inequality implies that for any $a \in (0,1)$
\begin{align}
\Pr[ F \ge (\expec F)^{a} ] \le \frac{\expec F}{(\expec F)^{a}} = (\expec F)^{1-a}.
\end{align}
This implies in particular that the variance of the purity,
which is at most $\expec \left[ F^2 \right]$, decreases exponentially.
In fact, this ensemble variance can be computed more explicitly.
We average the square of the post-measurement purity over both measurement outcomes and the unitaries $U$.
We get the same answer as in the post-selected case (see \cref{apnoise2} in \cref{SD}):
\begin{align} \label{noise2}
\expec \left[\frac{(\tr P\rho P\rho)^2}{(\tr P\rho)^3}\right]-\expec \left[2\frac{\tr P\rho P\rho}{\tr P \rho}\right]^2 &= \frac{4}{N} \left( \tr \rho^4-2(\tr \rho^3)(\tr \rho^2) + (\tr \rho^2)^3\right) + \ldots
\end{align}
so we obtain the same bounds as in the post-selected case:
for small $\tr \rho^2$ the noise is bounded by~$4 (\tr \rho^2)^2 / N$, 
whereas for a nearly pure state, $\tr \rho^2=1-\epsilon$ with small $\epsilon$, 
the noise is bounded by $4 (\tr \rho^2)^2 / N$.
This is consistent with the lack of noise seen in the red measurement curve in \cref{fignum}
at early and late times.

Second, the von Neumann entropy $S_{vN}$ also 
is exponentially small with probability exponentially close to~$1$.
To see this, we observe that for any state $\rho$ of impurity $F = 1 - \tr \rho^2$
there exists a pure state $\sigma$ within trace distance $\tfrac 12 \| \rho - \sigma \|_1 \le F$.
Then, the continuity of von Neumann entropy (the Fannes-Audenaert inequality~\cite{Audenaert2006})
implies that for any $F \le \tfrac 1 2$
\begin{align}
S_{vN}(\rho) \le F \log(N-1) - F \log F - (1-F) \log (1-F) = O( F \log(1/F) \log N).
\end{align}

In \cref{SDrank2} we explicitly work out the case of rank $2$ density matrices $\rho$, where all the terms in the various equations can be expressed in terms of the purity $\tr \rho^2$.

\section{Connection to Dyson Brownian motion}
\label{sec:Dyson}

Brownian motion is a stochastic process of a point $x(t)$ in $\mathbb R^d$
where for each (small) time step $\rd t$ the point is displaced by $x(t + \rd t) - x(t) = \rd x$
that is drawn from the Gaussian distribution with mean zero and variance $\rd t$.
It is well known that in the limit $\rd t \to 0$,
the density of the points evolves according to the heat equation.
Dyson Brownian motion~\cite{Dyson} is a Brownian motion in $\RR^{d^2}$ 
that is identified with the space of all Hermitian matrices.
More concretely, the update rule for a Hermitian matrix $X(t)$ is that
the increment $X(t+\rd t) - X(t)$ is drawn from the Gaussian unitary ensemble
where each independent matrix element is normalized to have variance $\rd t$.
The probability density $\sigma$ of the eigenspectra $\{ \lambda_1(t), \ldots, \lambda_d(t)\}$ of $X(t)$ 
follows the Dyson partial differential equation:
\begin{align}
\partial_t \sigma &= \mathsf D^\dagger \sigma,\nonumber \\
\mathsf D^\dagger \sigma,
&= 
- \sum_{a \neq b} \partial_a \frac{\sigma }{\lambda_a - \lambda_b} 
+ \frac 1 2 \sum_{a} \partial_a^2 \sigma.
\end{align}
The purpose of this section is to derive an analogous differential equation
for the dynamics of probabilistic ensembles of states according 
to the ``measurement'' setting defined in the introduction.%
\footnote{
The ``post-selection'' setting can be similarly handled,
and is left to interested readers.
}

We will derive the following equation that holds for an
\emph{arbitrary} real-valued smooth function $F : \lambda \mapsto F(\lambda)$
of the eigenspectrum $\{ \lambda_1, \ldots, \lambda_d \}$ of a state of rank $\le d \ll N$:
\begin{align} 
\partial_t \expec F 
&= \expec \mathfrak D F,\label{eq:pdeF} \\
\mathfrak D F &= \sum_{a,b:~a \neq b}^d (\partial_a F) \frac{ \lambda_a \lambda_b }{\lambda_a - \lambda_b}
+
\frac 1 2 \sum_{a,b}^d (\partial_a \partial_b F)\lambda_a \lambda_b \left(\delta_{ab} - \lambda_a - \lambda_b + \tr(\rho^2)\right) \nonumber
\end{align}
where the expectation is over the distribution of $\rho$ at the given time.
If desired, one may turn this into an equation for the probability density of eigenspectra
by applying the adjoint differential operator~$\mathfrak D^\dagger$.
To have a differential equation 
we are of course taking some limit of our discrete measurement process.
Here, we are taking $N$, the dimension of the total Hilbert space, to infinity
while setting the ``infinitesimal'' time step~$\rd t$ to be equal to~$\frac 1 N$.
In other words, \cref{eq:pdeF} governs the dynamics of the ensemble average of~$F$ 
accurately to order~$\frac 1 N$
when the initial state~$\rho$ has rank at most~$d \ll N$.

Before we start the derivation,
we note that the equation implies that 
if $F$ is a function of the sum of all eigenvalues $\sum_a \lambda_a = 1$,
then $\partial_t \expec F = 0$,
as expected.
Indeed, if $F(\lambda) = f(\sum_a \lambda_a)$
for some smooth real function $f$,
then $\partial_a F = f'$ and $\partial_a \partial_b = f''$
and the first sum vanishes because the summand is antisymmetric under $a \leftrightarrow b$
and the second sum becomes $(\tr \rho^2) - 2(\tr \rho^2)(\tr \rho) + (\tr \rho^2)(\tr \rho)^2 = 0$.

Note also that the calculation of the previous section is reproduced.
If $F(\lambda) = \lambda_1^2 + \lambda_2^2 + \cdots + \lambda_d^2$, 
is the purity,
then 
\begin{align}
\mathfrak D F 
&= 
\sum_{a<b} \left(\frac{2\lambda_a^2 \lambda_b}{\lambda_a - \lambda_b} 
+ \frac{2\lambda_b^2 \lambda_a}{\lambda_b - \lambda_a} \right)
+
\frac 1 2 
\sum_{a,b} 2\delta_{ab} \lambda_a \lambda_b \left( \delta_{ab} - \lambda_a - \lambda_b + \tr(\rho^2)\right)
\nonumber \\
&=
2\sum_{a<b} \lambda_a \lambda_b
+
\sum_a \lambda_a^2( 1- 2 \lambda_a + \tr(\rho^2))\\
&= \big(\sum_a \lambda_a \big)^2 - \sum_a \lambda_a^2 + (\tr \rho^2) - 2 (\tr \rho^3) + (\tr \rho^2)^2
\nonumber
\end{align}
which agrees with \cref{meas}.
Observe that \cref{mconv} is reproduced by the first term (the level repulsion term)
of~$\mathfrak D F$ in~\cref{eq:pdeF}.

\subsection{Derivation}
Here we derive \cref{eq:pdeF} assuming $d$ is a constant independent of $N$.  Put 
\begin{align}
\rd t = \frac 1 N.
\end{align}
For any given unitary $U$ and a (fixed) projector $P$ of rank $N/2$,
we find it convenient to introduce $M$ defined by
\begin{align}
\frac{I + M}{2} = UPU^\dagger \label{eq:defM}
\end{align}
where the first and second moments are%
\footnote{
These are computed by
\begin{align*}
\expec_U \left[ U \ket a \bra b U^\dagger \otimes U \ket c \bra d U^\dagger \right]
&= 
\sum_{s = \pm 1} \frac{ I + s \Sw }{2N (N + s)}( \delta_{ab}\delta_{cd} + s \delta_{ad}\delta_{bc} )
\\
\expec \bra x U^\dagger P U \ket y \bra z U^\dagger P U \ket w
&=
\tr\left[ \expec \left( U \ket y \bra x U^\dagger \otimes U \ket w \bra z U^\dagger \right) (P \otimes P) \right]
=
\frac{(N^2 -2)\delta_{xy}\delta_{wz} + N \delta_{xw}\delta_{yz}}{4(N^2 -1)}
\end{align*} where $I$ is the identity and $\Sw$ is the operator that swaps the two tensor factors.
}
\begin{align}
\expec_U M_{ab} &= 0,\nonumber \\
\expec_U M_{ab} M_{cd} &= \delta_{ad} \delta_{bc} \rd t + \cO( \rd t^2 ) .\label{eq:cov}
\end{align}
Observe that these moments agree with those of Gaussian unitary ensemble, rescaled by $\rd t$,
up to the leading order.

After the measurement of $\rho$ by $\{ U^\dagger P U, I- U^\dagger P U\}$
we have two outcomes whose spectra coincide with those of
\begin{align}
\rho'_M 
&= \frac{\sqrt{\rho} U^\dagger P U \sqrt{\rho}}{\tr( \sqrt \rho U^\dagger P U \sqrt \rho )} 
= \frac{ \rho + \sqrt\rho M \sqrt \rho }{1 + tr M \rho}, \nonumber\\
\rho'_{-M}
&= \frac{\sqrt{\rho} U^\dagger(I-P) U \sqrt{\rho}}{\tr( \sqrt \rho U^\dagger(I-P) U \sqrt \rho )} 
= \frac{ \rho - \sqrt\rho M \sqrt \rho }{1 - tr M \rho}.
\end{align}
Given $M$, we obtain $\rho'_M$ with probability $\frac 1 2 ( 1 + \tr M\rho)$ 
and $\rho'_{-M}$ with probability $\frac 1 2 (1 - \tr M\rho)$.
Thus, the expectation of $F(\rho')$ for any function $F$ over post-measurement states $\rho'$
is given by
\begin{align}
\int F(\rho'_M) \frac 1 2 (1 + \tr M \rho) \rd M + \int F(\rho'_{-M}) \frac 1 2 ( 1 - \tr M \rho) \rd M. \label{eq:avg-post-update}
\end{align}
Since the Haar measure on $U$ is invariant under the left multiplication by $V$ where $V^\dagger P V = I - P$,
we see that the random matrix $M$ has the same distribution as $-M$,
implying that the two terms of \cref{eq:avg-post-update} are the same.
We conclude that the post-measurement ensemble given $\rho$ is
\begin{align}
\{ ( \rho'_M, \rd M' ) \} \quad \text{ where } \quad \rd M' = (1+\tr M \rho)\rd M
\label{eq:pm-ensemble}
\end{align}
where $M$ is determined from $U$ by \cref{eq:defM}, and $\rd M$ is the induced measure.
Hence, $\int M_{ab} M_{cd} \rd M = \delta_{ad} \delta_{bc} \rd t$ by \cref{eq:cov}.
Let $\rho'_M$ have eigenvalues~$\lambda'_1,\ldots, \lambda'_d$.
By second order perturbation theory,
\begin{align}
\left(1+ \tr M \rho \right) \lambda'_a 
&=
\lambda_a 
+ \lambda_a M_{aa} 
+ \sum_{1 \le b \le d:~ b \neq a} \frac{ \lambda_a \lambda_b |M_{ab}|^2 }{\lambda_a - \lambda_b} 
+ \cO(\frac 1 {N^{3/2}})
\end{align}
where $1+ \tr M \rho$ in front of $\lambda'_a$ is to normalize $\lambda'$ so that $\sum_a \lambda'_a = 1$.
So the change in the eigenvalue is
\begin{align}
\rd \lambda_a &= \lambda'_a - \lambda_a = \frac{1}{1+ \tr M \rho} 
\underbrace{
\left(
\lambda_a M_{aa}
+ \sum_{1 \le b \le d:~ b \neq a} \frac{ \lambda_a \lambda_b |M_{ab}|^2}{\lambda_a - \lambda_b} - \lambda_a \tr M \rho
\right)
}_{\xi_a}
+ \cO(\rd t)^{3/2}
\end{align}

For an arbitrary smooth function $F$ from spectra to $\RR$,
we calculate the expected value of $F$ with respect to 
the post-measurement ensemble in \cref{eq:pm-ensemble} as follows.
That is, $\expec F = \int F \rd M'$.
We will ultimately want to take the expectation over a general ensemble,
but for the moment we are assuming that the pre-measurement ensemble 
is a Dirac delta distribution.
Abbreviate $\frac{\partial}{\partial x_a}F(x_1,\ldots,x_d)$ as $\partial_a F$.
We only keep terms up to order~$\rd t$.
\begin{align}
\expec \rd F(\lambda)
&=
\sum_{a=1}^d (\partial_a F) \expec \rd \lambda_a + \frac 1 2 \sum_{a,b}^d (\partial_a \partial_b F)\expec \rd \lambda_a \rd \lambda_b \nonumber \\
&=
\sum_a (\partial_a F) \int \frac{ \xi_a \rd M' }{1+ \tr M \rho}
+ 
\frac 1 2 \sum_{a,b}^d (\partial_a \partial_b F)
\int  \frac{\xi_a \xi_b \rd M'}{(1+ \tr M \rho)^2} \label{eq:EdF}\\
&=
\sum_a (\partial_a F) \int \xi_a \rd M
+ 
\frac 1 2 \sum_{a,b}^d (\partial_a \partial_b F)
\int \frac{\xi_a \xi_b \rd M}{1+ \tr M \rho} \nonumber
\end{align}
In the second term of the last line 
we may ignore the denominator because $\xi_a \xi_b$ is already $\cO(\rd t)$.
Using~\cref{eq:cov}, we arrive at~\cref{eq:pdeF} after integrating over a pre-measurement distribution,
and letting $\rd t \to 0$.

\subsection{(Un)Importance of Haar randomness}

In the above derivation, we assumed $d \ll N$ is fixed and used the following:
(i) the first and the second moments of $M$ is given by \cref{eq:cov},
and (ii) $\lim_{\rd t \to 0} \expec \|M\|^3/\rd t = 0$.%
\footnote{
That we have used nondegenerate perturbation theory is not an issue;
though we are not going to rigorously prove it,
this is fine because the level repulsion term separates the eigenvalues under the dynamics.
}
The first condition is obviously used,
and the second is to estimate the error term by Taylor's theorem.
(Since $F$ is a smooth function over a compact domain, the derivatives are all bounded.)
This implies that our differential equation~\cref{eq:pdeF}
holds even if the random unitary $U$ 
were, for example, a unitary $4$-design.
Indeed, by definition, any unitary $4$-design $U$ gives
the same expectation value for every quartic polynomial function of $M=M(U)$ 
as if~$U$ were Haar random.
By concentration of measure for a Haar random unitary~$V$ as we discussed in the previous section,
we know that any fourth moment of $M(V)$ is $\tilde \cO(\rd t^2)$.
Observe that $\| M\|^4 \le (\tr M^\dagger M)^2$ and the latter is a quartic polynomial in~$M$.
Since $\expec \| M \|^3 \le ( \expec \| M \|^4 )^{\frac 3 4}$ by concavity of $x \mapsto x^{\frac 3 4}$,
we see that $\expec_U \| M(U) \|^3 \le \tilde \cO(\rd t)^{3/2}$ for any unitary $4$-design $U$.

\section{Free Fermion Dynamics}
\label{secFFD}

We now specialize to the case of free fermion dynamics.  Consider the $2^n$ dimensional Fock space of $n$ fermionic modes, acted on by creation and annihilation operators $a_\mu^\dagger$ and $a_\mu$, $\mu=1,\ldots,n$.  In this section we will study {\emph{non-interacting}} measurement-unitary dynamics.  This is dynamics where the observables being measured are quadratic in the $a_\mu, a_\mu^\dagger$, and the unitaries being applied are exponentials of anti-Hermitian operators quadratic in the $a_\mu, a_\mu^\dagger$.  We will allow for non-particle conserving processes, and will find it useful to work with the Majorana operators $\gamma_{2\mu-1}=a_\mu+a_\mu^\dagger, \gamma_{2\mu}= i (a_\mu - a_\mu^\dagger)$, $\mu=1,\ldots,n$.


Our main result in this section is that, starting from the maximally mixed state, or, more generally, any mixed Gaussian state, a free fermion system purifies after $\sim n^2$ random free fermion measurements.  This is slower than e.g. a system in an area-law entanglement phase, which purifies in time $\sim n$, but much faster than a general interacting system, whose purification time scales like the many body Hilbert space dimension ($2^n$ in this case), as we shall see in the next section.  In the framework of Gullans and Huse \cite{Gullans}, a free fermion hybrid measurement-unitary circuit is thus always in the purifying phase, consistent with the intuition that a quantum error correcting code cannot dynamically emerge from free fermion dynamics.  We emphasize that we make no locality assumptions: the free fermion unitaries are fully random and thus non-local.

Since our evolution takes place entirely within the space of Gaussian states, let us recall some basic facts about these states, following Ref. \cite{Bravyi}.  A Gaussian state $\rho$ is a density matrix that, when written as a polynomial in the Majorana operators $\gamma_j$ with each $\gamma_j$ appearing with exponent $0$ or $1$ in each term, can be put into the form

\begin{align} 
    \frac{1}{2^n} \exp\left(\frac{i}{2}\theta^T M \theta \right) \nonumber
\end{align}
when the $\gamma_j$ are replaced with Grassmann numbers $\theta_j$.
Here $M$ is a real anti-symmetric $2n$-by-$2n$ matrix known as the correlation matrix.
We have
\begin{align}
    M_{ij}=\frac{i}{2} \Tr\, \rho \, [\gamma_i,\gamma_j]. \nonumber
\end{align}
Any such $M$ can be transformed by an $SO(2n)$ rotation $R$ into the following block-diagonal form:
\begin{align}
    M= R \, \bigoplus_{\mu=1}^n \begin{pmatrix}
      0 & \lambda_\mu \\
      -\lambda_\mu & 0
    \end{pmatrix} R^T \nonumber
\end{align}
where the $\lambda_\mu$ satisfy $-1 \leq \lambda_\mu \leq 1$ and are known as the Williamson eigenvalues of $M$.  Implementing the rotation $R$ on Fock space, the state $\rho$ is transformed into

\begin{align}
\rho_0 = \frac{1}{2^{n}} \prod_{\mu=1}^n \left(1+i \lambda_\mu \gamma_{2\mu-1}\gamma_{2\mu} \right). \nonumber
\end{align}
Since $i \gamma_{2\mu-1} \gamma_{2\mu}= a_{\mu} a_{\mu}^\dagger - a_{\mu}^\dagger a_{\mu}$ is the operator that measures the fermion parity of mode $\mu$, we see that $\rho_0$ is a tensor product state where each mode $\mu$ is independently filled or empty with probabilities $\frac{1}{2} (1-\lambda_\mu)$ and $\frac{1}{2}(1+\lambda_\mu)$ respectively.  The correlation matrix of $\rho_0$ is

\begin{align}
    M_0= \bigoplus_{\mu=1}^n \begin{pmatrix}
      0 & \lambda_\mu \\
      -\lambda_\mu & 0
    \end{pmatrix}. \nonumber
\end{align}

Before proceeding, it will also be useful to define the following function on Gaussian states:
\begin{align}
    S_{\mathrm{proxy}} (\rho)&=(\log 2)(n + \frac{1}{2} \Tr\,M^2) \nonumber \\
    &= (\log 2)(n - \Tr\,\mM^2)
\end{align}
where $\mM_{\mu \nu} = 2\, \Tr \, (\rho \, a_\mu a_\nu^\dagger) - \delta_{\mu\nu}$ is discussed below.
$S_{\mathrm{proxy}} (\rho)$ is a proxy for the second Renyi entanglement entropy 
$S_2(\rho)=n \log 2 - \frac{1}{2} \,\Tr \, \log (1-M^2)$ 
because it agrees with it on the maximally mixed state (where $M=0$) 
and all Gaussian pure states (where $M^2 = - 1$),
and in all other cases stays within order~$1$ constant multiples of $S_2$.
In the rest of this section we will establish bounds on the rate at which $S_{\mathrm{proxy}}$ decreases;
these will immediately translate to bounds on the second Renyi entropy.

\subsection{Particle number conserving dynamics} \label{conserv}

If we start with $\rho_0$ and apply dynamics that conserves $U(1)$ particle number, then the system moves through a restricted set of Gaussian states which are slightly easier to work with, and whose form we now derive.  A $U(1)$ particle number conserving unitary acts on the operator algebra by:
\begin{align}
    a_\mu &\rightarrow \mU_{\mu \nu}^* \, a_\nu \nonumber \\
    a_\mu^\dagger &\rightarrow \mU_{\mu \nu} \, a_\nu^\dagger, \nonumber
\end{align}
where $\mU$ is a unitary $n$-by-$n$ matrix.  Written in terms of Majoranas this is:

\begin{align} \label{Udef}
    \gamma_i \rightarrow O_{ij}\,\gamma_j 
\end{align}
where the $2n$-by-$2n$ orthogonal matrix $O$ describes the action of the unitary $\mU$ on ${\mathbb C}^{n}$ viewed as a real vector space ${\mathbb R}^{2n}$.  Explicitly, the $2$-by-$2$ block spanning rows $2\mu-1$ and $2\mu$ and columns $2\nu-1$ and $2\nu$ of $O$ ($\mu,\nu=1,\ldots, n$) is 
\begin{align} \label{specialO}
  \begin{pmatrix}
  \mathrm{Re} \, \mU_{\mu\nu} & -\rm{Im} \, \mU_{\mu\nu} \\
  \rm{Im} \, \mU_{\mu\nu} & \rm{Re} \, \mU_{\mu\nu} 
  \end{pmatrix}
\end{align}
Now consider a Gaussian state $\rho={\hat{U}}^\dagger \rho_0 {\hat{U}}$, where ${\hat{U}}$ is the Fock space operator that implements the action of $\mU$ in \cref{Udef}.  We have

\begin{align}
    \frac{i}{2} {\Tr}\,\rho \,[\gamma_j,\gamma_k] &= 
    \frac{i}{2} {\Tr}\,\rho_0\, {\hat{U}} [\gamma_j,\gamma_k] {\hat{U}}^\dagger \nonumber \\
    &= O_{jj'} O_{kk'} (M_0)_{j'k'} \nonumber \\
    &= \left(O M_0 O^T \right)_{jk}, \nonumber
\end{align}
so the correlation matrix of $\rho$ is $M=O M_0 O^T$.  
Again identifying ${\mathbb{R}}^{2n}$ with ${\mathbb C}^{n}$, 
we see that $M$
is the underlying real space action of the anti-Hermitian operator \mbox{$-i {\mathcal{M}}=-i \mU \mM_0 \mU^{\dagger}$},
where $\mM_0$ is the diagonal matrix with $\lambda_\mu$ on the diagonal.
We have \mbox{$\mM_{\mu \nu} = 2\, \Tr \, (\rho \, a_\mu a_\nu^\dagger) - \delta_{\mu\nu}$}.
This class of Gaussian states is precisely the class 
in which pairing correlations $\Tr \, (\rho \, a_\mu a_\nu)$ all vanish.

Let us now consider $\Delta S_{\mathrm{proxy}}$, 
the change in $S_{\mathrm{proxy}} (\rho)=(\log 2)(n - \Tr\,\mM^2)$, 
averaged over measurement outcomes, 
after measuring the occupation number of the first mode, 
i.e., the observable $i \gamma_1 \gamma_2$.  
In \cref{apfree} we show that 
\begin{align} \label{freeconserving}
\Delta S_{\mathrm{proxy}} &= - \frac{\log 2}{1-(\mM_{11})^2} (1-(\mM^2)_{11})^2 \leq 0
\end{align}
where $\mM_{11}$ is the entry in the first row and first column of $\mM$.
We thus see that, averaged over measurement outcomes the proxy entropy cannot increase, 
consistent with the same facts about the von Neumann and second Renyi entropies 
as proved in \cref{entineq}.

Now imagine that we have an ensemble of density matrices,
one that results for example from the application of several steps of a hybrid unitary-measurement circuit.
Letting the bar denote the average over this ensemble, we have:
\begin{align} \label{freefermionent}
\left|\overline{\Delta S_{\mathrm{proxy}}}\right| 
\geq (\log 2)\overline{(1-(\mM^2)_{11})^2} 
\geq (\log 2)\left(\overline{1-(\mM^2)_{11}}\right)^2
\end{align}
Now suppose that our circuit consists of Haar random free fermion unitaries 
interspersed with measurements of the first mode.
In this case, 
$\overline{1-(\mM^2)_{11}} = 1- \frac{1}{n} \overline{\Tr\, \mM^2} = \overline{S_{\mathrm{proxy}}} / (n \log 2)$.
Letting $s=S_{\mathrm{proxy}} / (n \log 2)$ be the density of the proxy entropy, we then have
\begin{align}
\left|\overline{\Delta s}\right| \geq \frac 1 n \overline{s}^2.
\end{align}

This equation roughly means that when the average proxy entropy density is $\overline{s}$, 
we learn at least $\sim {\overline{s}}^2$ about the system by measuring a single mode.  
It implies that $\overline{s(t)} \leq (1+t/n)^{-1}$, 
where $t$ measures the number of unitary-measurement steps taken.  
Note that due to the convexity properties above, 
this is a rigorous upper bound on $\overline{s(t)}$.  
We thus see that the system loses half of its entropy density in a time~$\sim n$, 
and its purity becomes of order~$1$ in a time~$\sim n^2$.
We expect these results to hold for a much more general class of free fermion unitaries, 
since the Haar random case intuitively corresponds to the situation where mixing is maximal.
Indeed, one generalization is a protocol where one alternates the application of an arbitrary free fermion unitary (not necessarily Haar random, and possibly different at each step) and measurement of each site with some nonzero probability.  In this case, an argument similar to the above shows that the proxy entropy decreases as $\sim 1/t$, implying that a mixed phase cannot be sustained in such a system.  As a consequence, an entanglement volume law phase cannot be sustained in such a system either.

\subsection{Particle number non-conserving dynamics}

Now let us perform the same calculation for a general Gaussian state $\rho$, with $2n$-by-$2n$ real anti-symmetric correlation matrix $M$ which may now contain non-zero pairing correlations.  To find the correlation matrix of the post-measurement state we use the techniques of Sec. VIII of Ref.~\cite{Bravyi}.  
Let~$K$ be the $2n$-by-$2n$ matrix $K_{pq} = \delta_{p1} \delta_{q2} - \delta_{p2} \delta_{q1}$.  Let $\alpha=M_{12}$, let $Q_{pq} = \delta_{p1} \delta_{q1} + \delta_{p2} \delta_{q2}$ be the projector on the first two basis vectors, and let $P=1-Q$.  Note that $QMQ=\alpha K$.  The probabilities of the first mode being empty and filled are $p_{\pm} = (1 \pm \alpha) /2$, and the post-measurement correlation matrix is

\begin{align}
M'_{\pm}=\pm K + P\left( M \pm \frac{1}{1\pm \alpha} MKM \right) P. \nonumber
\end{align}
We thus have

\begin{align}
p_+ \Tr\, (M'_+)^2 + p_- \Tr\, (M'_-)^2 = -2 + \Tr\, (PMP)^2 + \frac{1}{1-\alpha^2} \Tr\,(PMKMP)^2 \nonumber
\end{align}
On the other hand,

\begin{align}
\Tr\,M^2= \Tr\,(P+Q)M(P+Q)M = \Tr\, (PMP)^2 + 2 \Tr\, PMQMP - 2 \alpha^2 \nonumber
\end{align}
Thus

\begin{align} \label{DeltaSeq}
\Delta S_{\mathrm{proxy}} &= \frac{\log 2}{2} \left( p_+ \Tr\, (M'_+)^2 + p_- \Tr\, (M'_-)^2 - \Tr\, M^2 \right) \nonumber \\
&= -(1-\alpha^2) + \Tr\, KMPMK-\frac{1}{2(1-\alpha^2)} \Tr\, \left((KMPMK)K(KMPMK)K\right)
\end{align}
We note that the only non-zero entries of $KMPMK$ are in the upper left $2$ by $2$ block; these are

\begin{align}
\begin{pmatrix}
A & -C \\
-C & B 
\end{pmatrix}\nonumber
\end{align}
with $A=\sum_{j=3}^{2n} m_{2j}^2$, $B=\sum_{j=3}^{2n} m_{1j}^2$, and $C=\sum_{j=3}^{2n} m_{1j} m_{2j}$.  Let $x=(0, m_{12},m_{13}, \ldots, m_{1(2n)})$ and $y=(m_{21},0,m_{23},\ldots,m_{2(2n)})$ be the first two rows of $M$.  Note that $|x|^2=\alpha^2+B$ and $|y|^2 = \alpha^2+A$.  Using these facts, the expression in \cref{DeltaSeq} simplifies after some algebra to

\begin{align}
\Delta S_{\mathrm{proxy}} &= -\frac{\log 2}{1-\alpha^2}\left[ (1-|x|^2)(1-|y|^2) - (x \cdot y)^2 \right] \nonumber \\
&= -\frac{\log 2}{1-\alpha^2}\, \det\, Q (1+M^2) Q \nonumber
\end{align}
where we take the determinant of only the upper left $2$-by-$2$ block of $Q(1+M^2)Q$.
The matrix $1+M^2$ is positive semidefinite since its eigenvalues all lie between $0$ and $1$, 
so the same is true of $Q(1+M^2)Q$, 
and hence the determinant above is always non-negative.
This shows that the proxy entropy, averaged over measurement outcomes, always decreases.

Let us now start from the maximally mixed state and 
alternately apply $SO(2n)$ Haar-random unitaries and measurement operations.
We claim that this hybrid unitary-measurement circuit purifies in time $\sim n^2$.
To see this, we have to average the above formula for $\Delta S_{\mathrm{proxy}}$ 
over Gaussian states with with correlation matrices of the form $O M O^T$.
This requires computing averages of various quartic expressions in the entries of $O$, 
which we do using the Weingarten calculus.%
\footnote{
The average of any quartic polynomial in Haar random $O \in SO(2n)$ can be computed as follows.
The average $E = \expec_O O \ket a \bra b O^T \otimes O \ket c \bra d O^T$
commutes with $R \otimes R$ for any $R \in SO(2n)$.
Therefore~\cite{Collins2004}, we must have $E = xI + yS + z W$ 
for some $x,y,z \in \mathbb R$
where
$I = \sum_{j,k=1}^{2n} \ket{j,k}\bra{j,k}$,
$S = \sum_{j,k=1}^{2n} \ket{k,j}\bra{j,k}$,
and
$W = \sum_{j,k=1}^{2n} \ket{j,j}\bra{k,k}$.
Observe that $\Tr I = 4n^2$, $\Tr S = 2n$, and $\Tr W = 2n$.
Now, $\Tr E$, $\Tr ES$, and $\Tr EW$ can be computed directly,
determining the coefficients $x,y,z$.
}
We use the fact that for $j\neq k$ 
we have 
$\overline{O_{j1} O_{j1} O_{k2} O_{k2}}=1/(4n^2) + \mathcal O(n^{-4})$,
$\overline{O_{j1} O_{j2} O_{k1} O_{k2}}=-1/(8n^3) + \mathcal O(n^{-4})$,
$\overline{O_{j1} O_{j1} O_{j2} O_{j2}}=1/(4n^2) - 1/(4n^3) + \mathcal O(n^{-4})$.
The average of the determinant in the above equation can then be expressed in terms of the spectrum of $1+M^2$; in fact, all that enters is the sum of the eigenvalues, $\Tr\,(1+M^2)$, and the sum of the squares of the eigenvalues, $\Tr\, (1+M^2)^2$.  We have, to leading order in $1/n$,

\begin{align} \label{pairingbound}
|\overline{\Delta S_{\mathrm{proxy}}}| &=  \frac{\log 2}{4n^2} \left[ \left(\Tr\,(1+M^2)\right)^2 - \Tr\, (1+M^2)^2 \right] \nonumber \\
 &\geq  (\log 2) \left[ \left( \frac{\overline{S_{\mathrm{proxy}}}}{n \log 2} \right)^2 - \frac{1}{2n}\left( \frac{\overline{S_{\mathrm{proxy}}}}{n \log 2} \right)\right]
\end{align}
We thus again see that, as long as the proxy entropy is much greater than $1$, 
the proxy entropy density decreases at a rate at least as fast as the square of this density, 
leading to similar bounds as in the particle number conserving case. 
In particular, starting from the maximally mixed state, 
half of the proxy entropy is lost in time $\sim n$, 
and the purity becomes of order~$1$ in time~$\sim n^2$.
However, from the above bound we cannot determine that 
the state will eventually purify (i.e., that the purity will tend to~$1$),
only that it will reach purity of order~$1$.
Also, the bound in \cref{pairingbound} may not apply to protocols other than the Haar random case.
Indeed, in Ref.~\cite{Skinner} it is shown that a $1$-dimensional fermionic system 
in which random $i\gamma_j \gamma_{j+1}$ operators are measured (for both parities of $j$) 
lies in the bond percolation universality class, 
and there we expect that after $t$ measurements 
--- which corresponds to $1+1$d time $\sim n$ --- 
the entropy should be of order $n/t$.  Furthermore, it is relatively easy to construct a (still free fermion) generalization of this model,
whose associated statistical mechanical model is the loop model with crossings of Ref.~\cite{Nahum}.  At any point corresponding to
the `Goldstone phase' of the associated loop model, the entropy - which corresponds to the spanning number - for a fixed ratio of
$n/t$ scales as $\log n$.

\section{Discussion}
\label{discuss}
We have analyzed the many-body and free-fermion case.
If the Hilbert space dimension~$N$ in the many-body case equals $2^n$ for some system of~$n$ qubits,
then in the many-body case, the relaxation time is exponentially long in the number of degrees of freedom.
This contrasts strongly with the free-fermion case where the relaxation time is 
only polynomially long in the number of degrees of freedom 
--- more precisely, it is of order $n^2$.
It is an interesting question whether some intermediate behavior is possible.

Let us mention one further toy model.
Consider a system of $n$ qubits, with Pauli measurements and Clifford unitaries applied.  If we apply a sequence of measurements $Z_1,Z_2,Z_3,\ldots$, with no intervening unitaries, then the system purifies in linear time.  On the other hand, if we apply random Cliffords in between measurements of a fixed Pauli, it is easy to see that the purification time is exponential in $n$.  
In this process, the many-body state can be described 
as a stabilizer state with $\leq n$ linearly independent stabilizers.
Initially, the maximally mixed state has no stabilizers.
Given a state with $k$ linearly independent stabilizers, 
a further measurement of a product of Paulis will add a new stabilizer to the state 
(hence decreasing the entropy by $1$ bit) 
if the new stabilizer commutes with all the previous stabilizers and is linearly independent of them.
Without loss of generality we can fix the $k$ given stabilizers to be $Z_1,Z_2,\ldots,Z_k$.
If we measure a random product of Paulis, 
the probability that it commutes with the previous stabilizers is exponentially small in $k$.%
\footnote{
If we measure a product of Paulis generated 
by taking a fixed Pauli and conjugating it by a Clifford chosen uniformly from all Cliffords, 
the result is a random product of Paulis chosen uniformly from all such products 
subject to the condition that it is not equal to the identity.
For a quick analysis, 
it is simplest to instead consider the case that the random product is chosen uniformly from all Paulis, 
including the identity, in which case nothing is measured.
Then, the probability that the measurement commutes with existing stabilizers is exactly $2^{-k}$ 
and the probability that it is linearly independent given that it commutes is $1-4^{-(n-k)}$ 
since this is the probability that it is not equal to the identity on the remaining $n-k$ qubits.
}
Thus, the purification proceed monotonically: 
each stabilizer measurement can either reduce the entropy by $1$ bit or leave it unchanged.
However, the purification time indeed is exponential in $n$.

Our results and this toy model suggest that 
there might be a dichotomy between polynomial and exponential relaxation times.  
However, perhaps richer relaxation behavior may be observed at a phase transition between these two possibilities.

\acknowledgments
LF acknowledges the support of NSF DMR 1939864, and useful conversations with Matthew P. A. Fisher.

\appendix

\section{Appendix: Entropy Inequalities}
\label{entineq}
Here we prove that entropy cannot increase on average after a measurement 
and square-root purity cannot decrease on average after a measurement.

\begin{lemma}
Given {\it any} density matrix $\rho$, and any measurement (either a projective measurement or more generally a POVM), the average, over measurement outcomes, of the von Neumann entropy of the state after measurement, cannot increase.
\begin{proof}
Note first that since a POVM can be implemented by a projective measurement in a larger Hilbert space, it suffices to consider the case of projective measurements.
Consider a tripartite system, with three subsystems $A,B,C$.  Choose subsystem $A$ have have reduced density matrix $\rho_A=\rho$ and choose $C$ to be an arbitrary purification of $\rho$, i.e., the density matrix $\rho_{AC}$ on $AC$ is a pure state.  We use $B$ as a register for the measurement: initially, $B$ is in some fixed state $|0\rangle$, and then some unitary on $AB$ is applied so that the state of $B$ in some given basis records the outcome of the measurement.  This unitary is chosen in the obvious way so that if $\rho$ is in the range of any of the projectors defining the projective measurement, then the reduced density matrix on $A$ is left unchanged.
We write $\tau$ for the density matrix of the tripartite system after this unitary is applied.

Then, the entropy of $A$ after measurement, averaged over measurement outcomes, is equal to
$S(\tau_A)-S(\tau_B)$.  Since $C$ still purifies $AB$, this is equal to $S(\tau_{BC})-S(\tau_B)$.  By subadditivity, this is $\leq S(\tau_C)=S(\rho)$.
\end{proof}
\end{lemma}

We now prove a similar result for the square-root of the purity of a quantum state.  
Then,
\begin{lemma}
Given {\it any} density matrix $\rho$, and any measurement (either a projective measurement or more generally a POVM), the average, over measurement outcomes, of the square-root purity of the state after measurement, cannot decrease.
\end{lemma}
\begin{proof}
As before, consider only projective measurements.  
Label measurement outcomes by an index~$i$ and regard~$\rho$ as a block matrix 
with row and column blocks labeled by this index $i$.
Let $\rho_{ij}$ denote the submatrix in the $i$-th row and $j$-th column block.
Each normalized block $\sigma_i = \rho_{ii} / \tr \rho_{ii}$ is a post-measurement state,
which comes with probability $p_i = \tr \rho_{ii}$.
Hence, our goal is to show
\[
\sqrt{\tr(\rho^2)} \le \sum_i \sqrt{\tr( \rho_{ii} ^2 )} = \sum_i p_i \sqrt{\tr(\sigma_i^2)} .
\]
This inequality follows from the case where the index $i$ assumes only two values
since one can consider coarser blocking of~$\rho$ and subdivide it.
So, assume that $\rho$ consists of $A = \rho_{11}, B = \rho_{12} = \rho_{21}^\dag$, and $D = \rho_{22}$.
Squaring the inequality, 
we see that our inequality is equivalent to
$\tr(B^\dag B) \le \sqrt{\tr(A^2)}\sqrt{\tr(D^2)}$.
Let us work in the basis where $B$ is diagonal 
by taking singular value decomposition of~$B$.
(In this basis, $A$ and $D$ are not necessarily diagonal.)
Since $\rho$ is positive semidefinite,
the determinant of any principal $2$-by-$2$ block must be nonnegative.
In particular, $A_{jj}D_{jj} - (B_{jj})^2 \ge 0$.
Summing over $j$, 
we have $\tr(B^\dag B) = \sum_j (B_{jj})^2 \le \sum_j A_{jj} D_{jj}$.
By Cauchy--Schwarz, $\sum_j A_{jj} D_{jj} \le \sqrt{\sum_j (A_{jj})^2}\sqrt{\sum_j (D_{jj})^2} \le \sqrt{\tr(A^2)}\sqrt{\tr(D^2)}$.
This completes the proof.
\end{proof}

\section{Appendix: Schwinger-Dyson computations}
\label{SD}


Consider a product $X$ of $N$ by $N$ complex matrices that includes some number of instances of a unitary matrix $U$ and an equal number of instances of its adjoint.  Let $M$ be an arbitrary complex $N$ by $N$ matrix, and consider the expression $\expec[\tr((M+M^\dagger)X)]$ where $E$ means averaging over all $U$ in $U(N)$ with respect to Haar measure.  This expression is invariant under a re-parametrization of $U$: $U\rightarrow e^{i \epsilon (M+M^\dagger)} U \approx (1+i\epsilon (M+M^\dagger))U$, where $\epsilon$ is infinitesimal.  Expanding to linear order in $\epsilon$ we get several terms summing to $0$.  Now multiply each term by $\exp(-\tr(M^\dagger M) / 2)$ and integrate over $M$.  The result is again several terms summing to $0$, with each term a product of traces of the original matrices appearing in $X$.  This is an example of a Schwinger-Dyson equation.  More generally, $X$ can also include additional factors of traces of products of matrices containing equal numbers of $U$'s and $U^\dagger$'s.

The following type of expression appears a lot when we perform the above steps:

\begin{align} \label{apeq1}
\int dM  &\left[ \tr (M+M^\dagger)A(M+M^\dagger)B\right] \exp\left(-\frac{1}{2}\tr MM^\dagger\right) = \\ \nonumber
& 2 \int dM \left[\tr MA M^\dagger B\right] \exp\left(-\frac{1}{2}\tr MM^\dagger\right) = c (\tr A) (\tr B)
\end{align}
where $c$ is an unimportant order $1$ constant.  Another useful expression which appears is

\begin{align} \label{apeq2}
\int dM  &\left[\tr (M+M^\dagger) A \right]  \left[\tr (M+M^\dagger)B \right] \exp\left(-\frac{1}{2}\tr MM^\dagger\right) = \\ \nonumber
& 2 \int dM \left[\tr M A \right] \left[\tr M^\dagger B \right] \exp\left(-\frac{1}{2}\tr MM^\dagger\right) = c \tr AB
\end{align}
Let us now use these to derive the Schwinger Dyson equations that are used in the paper.  Recall that $P = UP_0 U^\dagger$.  Let us define $H=M+M^\dagger$.  To simplify notation, let us use round brackets around operators to denote trace.  Let us start with $X=P \rho$.  Then

\begin{align}
(HP\rho) &\rightarrow (H(1+i\epsilon H) P (1-i\epsilon H)\rho) \nonumber \\ 
&= (HP\rho) + i\epsilon \left[(H H P \rho) - (HPH\rho) \right] \nonumber
\end{align}
and the sum of terms linear in $\epsilon$ will vanish under the expectation value (integration over $U$).  Let us now multiply these terms by $\exp\left(-\frac{1}{2}\tr MM^\dagger\right)$ and integrate over $M$.  The result, using \cref{apeq1}, is $\expec[N(P \rho)-(P)(\rho)]=0$, which, using $(P)=N/2$, amounts to $\expec[(P \rho)]=1/2$.  The same computation with $\rho$ replaced by $\rho^2$ gives $\expec[(P \rho^2)]=(\rho^2)/2$.

Now let us take $X=P \rho P \rho$.  Then
\begin{align}
(HP\rho P\rho) &\rightarrow (H(1+i\epsilon H) P (1-i\epsilon H) \rho (1+i \epsilon H) P (1-i \epsilon H) \rho) \nonumber \\
&= (HP\rho P\rho)+ i \epsilon \left[ (HHP\rho P\rho)-(HPH\rho P\rho) + (HP\rho HP\rho) + (HP\rho PH\rho)\right] \nonumber
\end{align}
and the sum of terms linear in $\epsilon$ will vanish under the expectation value (integration over $U$).  Let us now multiply these terms by $\exp\left(-\frac{1}{2}\tr MM^\dagger\right)$ and integrate over $M$.  The result, using \cref{apeq1}, is

\begin{align}
\expec[N(P\rho P\rho) - (P) (\rho P \rho) + (P\rho)(P\rho) - (P\rho P)(\rho)]=0. \nonumber
\end{align}
Using $(P)=N/2$, $(\rho)=1$, $\expec[(\rho P \rho)] = \expec[(P \rho^2)] = (\rho^2)/2$, this simplifies to 

\begin{align}\label{newap1}
\expec[(P\rho P\rho)]=(\rho^2)/4+\frac{1}{N} \left(-\expec[(P \rho) (P\rho)]+1/2\right).
\end{align}
Now take $X= P\rho (P \rho)$.  Then
\begin{align}
(HP\rho) (P\rho) &\rightarrow (H(1+i\epsilon H)P(1-i\epsilon H) \rho)((1+i\epsilon H) P (1-i\epsilon H) \rho) \nonumber \\
&= (HP\rho) (P\rho) + i \epsilon \left[(HHP\rho)(P\rho)-(HPH\rho)(P\rho)+(HP\rho)(HP\rho)-(HP\rho)(PH\rho) \right] \nonumber
\end{align}
Again, multiplying by $\exp\left(-\frac{1}{2}\tr MM^\dagger\right)$ and integrating over $M$ gives, using  \cref{apeq1,apeq2}:
\begin{align}
\expec[N(P\rho)(P\rho)-(P)(\rho)(P\rho)+(P\rho P\rho)-(P\rho \rho P)]=0. \nonumber
\end{align}
Using $(P)=N/2$, $(\rho)=1$, $(P\rho)=1/2$, $(P\rho^2 P) = (P\rho^2)=(\rho^2)/2$ this becomes 
\begin{align}\label{newap2}
\expec[(P\rho)(P\rho)]=1/4 + \frac{1}{N}\left( - (P\rho P\rho)+ (\rho^2)/2 \right)
\end{align}

\subsection{Computation of the expected value of the purity}

Let us derive \cref{postsel} of \cref{secMBD}.  
There are three terms to compute, given in \cref{betterseries}. 

\subsubsection*{\it First term---}

Inserting \cref{newap2} into \cref{newap1} and keeping only terms up to order $N^{-1}$ we obtain the first term:
\begin{align} \label{newap3}
\expec[(P\rho P\rho)]=(\rho^2)/4+\frac{1}{4N} + O(N^{-2})
\end{align}

\subsubsection*{\it Second term---}

The second term is $\expec[(P\rho P\rho)\delta]= \expec[(P\rho P\rho)((P\rho)-1/2)]$.  To compute this, let $X= P\rho P\rho (P\rho)$.  Then

\begin{align}
(HP\rho P\rho)&(P\rho) \rightarrow (HP\rho P\rho)(P\rho) + i\epsilon\left[ (HHP\rho P\rho)(P\rho)-(HPH\rho P \rho)(P\rho)+(HP\rho HP \rho)(P\rho) \right] \nonumber \\
& + i \epsilon \left[ - (HP\rho PH \rho) (P \rho) + (HP\rho P \rho)(HP \rho)-(HP\rho P\rho)(PH \rho) \right] \nonumber
\end{align}
Multiplying by $\exp\left(-\frac{1}{2}\tr MM^\dagger\right)$ and integrating over $M$ gives, using \cref{apeq1,apeq2}:

\begin{align}
\expec[N(P\rho P\rho)(P\rho) - (P)(\rho P \rho)(P \rho)+(P\rho)(P\rho)(P\rho)-(P\rho P)(\rho)(P\rho)+ (P\rho P \rho P\rho)-(P\rho P \rho \rho P)]=0. \nonumber
\end{align}
Using $(P)=N/2$, $(\rho)=1$ this can be re-written as:

\begin{align}
\expec[(P\rho P\rho)(P\rho)]=\frac{1}{2} (P\rho^2)(P\rho)+\frac{1}{N}\left[\frac{1}{8} + \frac{1}{8}(\rho^3)\right] + O(N^{-2})\nonumber
\end{align}
Now, we only care about the terms in brackets to zeroth order in $1/N$, so we can repeatedly use Schwinger-Dyson equations to turn each $P$ in each term in brackets to a $1/2$.   Note that the trace of a product of an arbitrary number of $P$'s and various powers of $\rho$ is always bounded by $1$, as can be seen by working in the eigenbasis of $\rho$, inserting complete sets of states between all the terms, and using the fact that $\langle i | P | j \rangle \leq 1$ for any unit vectors $|i\rangle, |j\rangle$.  Thus we can get uniform bounds on the terms at order $N^{-k}$, $k\geq 1$, and hence are justified in dropping them.  The result is:
\begin{align} \label{newap4}
\expec[(P\rho P\rho)(P\rho)]=\frac{1}{2} (P\rho^2)(P\rho)+\frac{1}{N}\left[-(P\rho)^3+(P\rho)^2-(P\rho P\rho P\rho)+(P\rho P\rho^2)\right]
\end{align}
Now let $X= P\rho^2 (P\rho)$.  Then
\begin{align}
(HP\rho^2)&(P\rho) \rightarrow (HP\rho^2)(P\rho) + i\epsilon\left[(HHP\rho^2)(P\rho)-(HPH\rho^2)(P\rho)+(HP\rho^2)(HP\rho)-(HP\rho^2)(PH\rho) \right] \nonumber
\end{align}
Multiplying by $\exp\left(-\frac{1}{2}\tr MM^\dagger\right)$ and integrating over $M$ gives, using \cref{apeq1,apeq2}:

\begin{align}
\expec[N(P\rho^2)(P\rho) - (P)(\rho^2)(P \rho)+(P \rho^2 P \rho)-(P\rho^3 P)]=0 \nonumber
\end{align}
which can be rearranged to

\begin{align} 
\expec[(P\rho^2)(P\rho)] = \frac{1}{4}(\rho^2) + \frac{1}{N}\left[-(P\rho^2 P\rho) + \frac{1}{2} (\rho^3) \right] \nonumber
\end{align}
Again replacing all the $P$'s with $1/2$'s in the term in brackets gives
\begin{align}
\expec[(P\rho^2)(P\rho)] = \frac{1}{4}(\rho^2) + \frac{1}{4N} (\rho^3) + O(N^{-2})\nonumber
\end{align}
Inserting this into \cref{newap4} we obtain 
\begin{align}
\expec[(P\rho P\rho)(P\rho)] = \frac{1}{8} (\rho^2) + \frac{1}{N}\left[ \frac{1}{4}(\rho^3)+\frac{1}{8}\right].\nonumber
\end{align}
Subtracting $\expec[(P\rho)(P\rho)/2]= (\rho^2)/8 + 1/(8N)$ (see \cref{newap3}) we then obtain
\begin{align} \label{newap5}
\expec[(P\rho P\rho)\delta] = \frac{1}{4N} (\rho^3)
\end{align}

\subsubsection*{\it Third Term---}

The third term is $\tr(P \rho P \rho) \delta^2$.  To leading order in $N$, this is $\expec[(P\rho P\rho)] \cdot \expec[\delta^2]=\frac{1}{4} (\rho^2) \cdot \expec[\delta^2]$, using the same arguments as in the calculation of the second term.  We have $\expec[\delta^2]=(1/2)N^{-1}(\rho^2)-N^{-1}\expec[(P\rho P\rho)] = (1/4)N^{-1} (\rho^2) + O(N^{-2})$ using \cref{delsquared,newap3}.  So

\begin{align} \label{newap6}
\expec[\tr(P \rho P \rho) \delta^2]=\frac{1}{16N}(\rho^2)^2
\end{align}

Thus, including the first three terms in the series we have, using \cref{newap3,newap5,newap6}:
\begin{align}\label{derivedsum}
\expec \left[\frac{\tr P \rho P \rho}{\tr(P \rho)^2}\right]&=\expec \left[4\tr P\rho P\rho-16(\tr P \rho P \rho) \delta + 48 (\tr P\rho P\rho)\delta^2\right] \\
&=(\rho^2)+N^{-1}-4N^{-1}(\rho^3)+3N^{-1}(\rho^2)^2+\ldots
\end{align}
which is \cref{postsel} of \cref{secMBD}.  These same three terms enter into \cref{meas}:

\begin{align}\label{derivedmeas}
\expec \left[2 \frac{\tr P\rho P\rho}{\tr P\rho}\right] &=\tr \rho^2+\frac{1}{N}\left[1-2(\tr \rho^3)+(\tr \rho^2)^2\right] + \ldots
\end{align}

\subsection{Schwinger-Dyson computation of the noise term}

We now derive the noise for the post-selected and measurement cases, \cref{noise1,noise2} in \cref{secMBD}.  First let us treat the post-selected case, and derive \cref{noise1}.  We have:

\begin{align} \label{noisesum}
\expec \left[\frac{(P\rho P\rho)^2}{(P\rho)^4}\right] &= \expec \left[\frac{16(P\rho P\rho)^2}{(1+2\delta)^4}\right] \nonumber \\
&= \expec \left[16(P\rho P\rho)^2 \left( 1-8 \delta + 40 \delta^2 + \ldots \right) \right]
\end{align}
so we just need to compute the expectation values of $(P\rho P\rho)^2$, $(P\rho P\rho)^2 \delta$, and $(P\rho P\rho)^2 \delta^2$.  We use the same Schwinger-Dyson methods as above, but give fewer details here.

\subsubsection*{\it Computing the first term: $\expec[(P\rho P\rho)^2]$}

First take $X=(P\rho P\rho)(P\rho P\rho)$.  We obtain:

\begin{align} \label{SD1}
(P\rho P\rho)(P\rho P \rho) &= \frac{1}{2}(P\rho^2)(P\rho P\rho) \nonumber \\ &+ \frac{1}{N}[-(P\rho)^2(P\rho P\rho)+(P\rho)(P\rho P\rho)-2(P\rho P\rho P\rho P\rho)+2(P\rho P\rho^2 P\rho)]
\end{align}
For the terms with coefficient $N^{-1}$ above, we only need to compute them to order $N^0$.  This can be done by repeated use of Schwinger-Dyson equations in which only two terms are of order $N$: essentially, these equations allow us to repeatedly replace $P$ with $\frac{1}{2}$.  We obtain, to order $N^0$:

\begin{align}
\expec[(P\rho)^2(P\rho P\rho)] &= \frac{1}{16} (\rho^2) \nonumber \\
\expec[(P\rho)(P\rho P\rho)] &= \frac{1}{8} (\rho^2) \nonumber \\
\expec[(P\rho P\rho P\rho P\rho)] &= \frac{1}{16} (\rho^4) \nonumber \\
\expec[(P\rho P\rho P\rho^2)] &= \frac{1}{8}(\rho^4) \nonumber
\end{align}
and when we plug these into \cref{SD1} get that to order $N^{-1}$,

\begin{align} \label{SD3}
\expec[(P\rho P\rho)(P\rho P\rho)]= \frac{1}{2}\expec[(P\rho^2)(P\rho P\rho)] + \frac{1}{N}\left(\frac{1}{16}(\rho^2) + \frac{1}{8}(\rho^4)\right)
\end{align}
Now we take $X=(P\rho^2)(P\rho P\rho)$.  We obtain

\begin{align} \label{SD2}
(P\rho^2)(P\rho P\rho) = \frac{1}{2}(\rho^2)(P\rho P\rho) + \frac{1}{N}\left( -2(P\rho^2 P\rho P\rho)+2(P\rho^3 P\rho)\right)
\end{align}
Again, to order $N^0$ we have:

\begin{align}
\expec[(P\rho^2 P\rho P\rho)] &= \frac{1}{8} (\rho^4) \nonumber \\
\expec[(P\rho^3 P\rho)] &= \frac{1}{4} (\rho^4) \nonumber
\end{align}
so plugging into \cref{SD2} we obtain to order $N^{-1}$:
\begin{align} \label{SD4}
\expec[(P\rho^2)(P\rho P\rho)] = \frac{1}{2} \expec[(\rho^2)(P\rho P\rho)] + \frac{1}{N}\left( \frac{1}{4}(\rho^4) \right).
\end{align}
Taking $X= (P\rho P\rho)(\rho^2)$ and performing similar steps gives, to order $N^{-1}$:
\begin{align} \label{SD5}
\expec[(P\rho P\rho)(\rho^2)] = \frac{1}{4} (\rho^2)^2 + \frac{1}{4N} (\rho^2).
\end{align}
Combining \cref{SD3,SD4,SD5} we obtain that to order $N^{-1}$,

\begin{align} \label{SD6}
\expec[(P\rho P\rho)^2]=\frac{1}{16}(\rho^2)^2 + N^{-1} \left(\frac{1}{8} (\rho^2) + \frac{1}{4} (\rho^4)\right)
\end{align}

\subsubsection*{\it Computing the second term: $\expec[(P\rho P\rho)^2 \delta]$}

First take $X=(P\rho)(P\rho P\rho)^2$.  The Schwinger-Dyson equation to order $N^{-1}$ is:

\begin{align}
\expec[(P\rho)(P\rho P\rho)^2] &= \frac{1}{2}\expec[(P\rho P\rho)^2 + \frac{4}{N}\left( (P\rho^2 P\rho)(P\rho P\rho) - (P\rho P\rho P\rho)(P\rho P\rho)\right)] \nonumber \\
&=  \frac{1}{2}\expec[(P\rho P\rho)^2] + \frac{4}{N}\left( \frac{1}{4}(\rho^3)\frac{1}{4}(\rho^2) -\frac{1}{8} (\rho^3) \frac{1}{4} (\rho^2)\right)\nonumber \\
&= \frac{1}{2}\expec[(P\rho P\rho)^2] + \frac{1}{N} \left( \frac{1}{8} (\rho^3) (\rho^2) \right).\nonumber
\end{align}
Recalling that $\delta=(P\rho)-\frac{1}{2}$, we therefore have that to order $N^{-1}$,
\begin{align} \label{SD7}
\expec[(P\rho P\rho)^2 \delta] =  \frac{1}{8N}(\rho^3)(\rho^2).
\end{align}

\subsubsection*{\it Computing the third term: $\expec[(P\rho P\rho)^2 \delta^2]$}

To order $N^{-1}$, $\expec[(P\rho P\rho)^2 \delta^2]$ is simply the product of $\expec[(P\rho P\rho)^2]$ and $\expec[\delta^2] = \frac{1}{4} N^{-1} (\rho^2)$:

\begin{align} \label{SD8}
\expec[(P\rho P\rho)^2 \delta^2] = \frac{1}{16}(\rho^2)^2 \frac{1}{4} N^{-1} (\rho^2) = \frac{1}{64} N^{-1} (\rho^2)^3
\end{align}
Inserting \cref{SD6,SD7,SD8} into \cref{noisesum} we therefore see that, to order $N^{-1}$

\begin{align} 
\expec \left[\frac{(P\rho P\rho)^2}{(P\rho)^4}\right] &= (\rho^2)^2 + \frac{1}{N} \left[ 2(\rho^2) + 4(\rho^4) - 16(\rho^3)(\rho^2) + 10(\rho^2)^3 \right]
\end{align}
On the other hand, using our previous computation we have that, to order $N^{-1}$:

\begin{align}
\expec \left[\frac{P\rho P\rho}{(P \rho)^2}\right]^2 = (\rho^2)^2 + \frac{2}{N}(\rho^2)(1-4(\rho^3)+3(\rho^2)^2)
\end{align}
Subtracting these two equations we therefore see that to order $N^{-1}$

\begin{align} \label{apnoise1}
\expec[\frac{(P\rho P\rho)^2}{(P\rho)^4}]-\expec[\frac{(P\rho P\rho)}{(P \rho)^2}]^2 &= \frac{4}{N} \left( (\rho^4)-2(\rho^3)(\rho^2) + (\rho^2)^3\right)
\end{align}
which is just \cref{noise1} in \cref{secMBD}.

Now let us do the computation in the case of measurement.  Here we want to average over both the measurement outcomes and the unitaries $U$.  Thus we want to compute

\begin{align}
E &\left[ (P\rho) \frac{(P\rho P\rho)^2}{(P\rho)^4}+ ((I-P)\rho) \frac{((I-P)\rho(I-P)\rho)^2}{((I-P)\rho)^4}\right] 
- \expec \left[2 \frac{(P\rho P\rho)}{(P\rho)} \right]^2= \nonumber \\
& \expec \left[ 2 \frac{(P\rho P\rho)^2}{(P\rho)^3} \right] - \expec \left[2 \frac{(P\rho P\rho)}{(P\rho)} \right]^2
\end{align}
Writing the denominator $(P\rho) = \frac{1}{2} + \delta$ and expanding in $\delta$ as usual, we obtain
\begin{align}
\expec \left[ 2 \frac{(P\rho P\rho)^2}{(P\rho)^3} \right] = 16 (P\rho P\rho)^2 \left[1-6\delta + 24\delta^2 + O(\delta^3)\right] \nonumber
\end{align}
By the same argument as in \cref{lemma1} the $O(\delta^3)$ terms add up to an error that is asymptotically smaller than $1/N$, so we just need to compute the first 3 terms.  Using \cref{SD6,SD7,SD8} these add up to

\begin{align}
\expec \left[ 2 \frac{(P\rho P\rho)^2}{(P\rho)^3} \right] = (\rho^2)^2+ \frac{1}{N} \left[2 (\rho^2)+ 4(\rho^4)-12(\rho^2)(\rho^3)+ 6(\rho^2)^3 \right] \nonumber
\end{align}
On the other hand, squaring \cref{meas} gives, to $O(N^{-1})$:
\begin{align}
\expec \left[2 \frac{(P\rho P\rho)}{(P\rho)} \right]^2 = (\rho^2)^2+ \frac{1}{N}\left[2(\rho^2) - 4(\rho^2)(\rho^3)+2(\rho^2)^3\right] \nonumber
\end{align}
Subtracting these two gives
\begin{align} \label{apnoise2}
\expec \left[ 2 \frac{(P\rho P\rho)^2}{(P\rho)^3} \right] - \expec \left[2 \frac{(P\rho P\rho)}{(P\rho)} \right]^2 = \frac{4}{N}\left[(\rho^4)-2(\rho^2)(\rho^3)+(\rho^2)^3\right]
\end{align}

Let us make some comments about higher order contributions to this.  Note that for $\rho = \frac{1}{N} {\bf{1}}$, all of the terms in the above expression are $O(N^{-4})$.  One could ask if there are any $O(N^{-2})$ contributions in this case.  The $N^{-2}$ term will look similar to the above, with a sum of products of traces of powers of $\rho$.  By noting that the right hand side has to vanish identically for $\rho = \frac{1}{N} {\bf{1}}$ (after summing the whole $1/N$ expansion), we see that, among the $N^{-2}$ terms, there cannot be a constant piece.  Hence we conclude that for $\rho$ close to $\frac{1}{N} {\bf{1}}$ (i.e. after a few applications of random $P$'s) the noise will be $O(N^{-3})$.

Let us now compute the noise term in \cref{apnoise2} in the case of a nearly pure state, where $\tr \rho^2 = 1-\epsilon$ with $\epsilon$ small.  We will compute this to second order in $\epsilon$.  In this case $\rho$ must have one large eigenvalue $1-\eta$ ($\eta \ll 1$) and $N-1$ small eigenvalues $\delta_1,\ldots,\delta_{N-1}$, with $\eta = \delta_1 + \ldots+\delta_{N-1}$.  We have $\epsilon = 2 \eta - \eta^2 - (\delta_1^2+\ldots+\delta_{N-1}^2)$, which can be inverted to $O(\epsilon^2)$ to give $\eta=\epsilon/2 + \epsilon^2/8+(\delta_1^2+\ldots+\delta_{N-1}^2)/2$.  Therefore, to order $\epsilon^2$,

\begin{align}
\tr \rho^3 = 1-\frac{3}{2}\epsilon + \frac{3}{8}\epsilon^2 - \frac{3}{2}\left(\delta_1^2+\ldots+\delta_{N-1}^2\right) \nonumber
\end{align}
and since $\delta_1^2 + \ldots + \delta_{N-1}^2 \leq \eta^2 = \epsilon^2 / 4$ to order $\epsilon^2$, we see that 

\begin{align}
1-\frac{3}{2}\epsilon \leq \tr \rho^3 \leq 1- \frac{3}{2}\epsilon + \frac{3}{8} \epsilon^2 \nonumber
\end{align}
A similar computation to order $\epsilon^2$ shows
\begin{align}
\tr \rho^4 = 1-2\epsilon+\epsilon^2 - 2\left(\delta_1^2+\ldots+\delta_{N-1}^2\right) \nonumber
\end{align}
so that
\begin{align}
1-2\epsilon + \frac{1}{2}\epsilon^2 \leq \tr \rho^4 \leq 1-2\epsilon+\epsilon^2 \nonumber
\end{align}
Plugging the upper (lower) bounds into the positive (negative) terms in \cref{apnoise2} shows that, to order $\epsilon^2$,

\begin{align} \label{noisebound}
\expec \left[ 2 \frac{(P\rho P\rho)^2}{(P\rho)^3} \right] - \expec \left[2 \frac{(P\rho P\rho)}{(P\rho)} \right]^2 \leq \frac{4}{N} \epsilon^2
\end{align}

\subsection{Rank 2 case} \label{SDrank2}

When $\rho$ has rank $2$, the various Schwinger-Dyson equations close on themselves.  Indeed, in that case if the eigenvalues of $\rho$ are $x, 1-x$, $0\leq x \leq 1$, then $(\rho^2)=1-2x+x^2$, and

\begin{align}
(\rho^3)=x^3+(1-x)^3=1-3x+3x^2 = \frac{3}{2}(\rho^2)-\frac{1}{2} \nonumber
\end{align}
Therefore, for the case of measurement,

\begin{align}
\expec \left[ \frac{(P\rho P\rho)}{(P\rho)^2} \right] &= (\rho^2) + N^{-1}\left[ 1-4(\rho^3)+(\rho^2)^2 \right]  \nonumber\\
&= (\rho^2) + N^{-1} [(\rho^2)-1][(\rho^2)-2]\nonumber
\end{align}
The resulting time evolution is:

\begin{align}
(\rho^2)= 1 - \frac{1}{3\exp\left(t/N\right) - 1}\nonumber
\end{align}
This can only be proved to be valid for $t \ll N$ however.  Using $(\rho^4) = (\rho^2)^2/2 + (\rho^2) - 1/2$, the $O(N^{-1})$ noise term, from \cref{apnoise2}, is
\begin{align}
\expec \left[ 2 \frac{(P\rho P\rho)^2}{(P\rho)^3} \right] - \expec \left[2 \frac{(P\rho P\rho)}{(P\rho)} \right]^2 = 4N^{-1}\left[ ((\rho^2)-1)^2 ((\rho^2)-\frac{1}{2}) \right]\nonumber
\end{align}
so the noise is very small at the lower end $(\rho^2)=1/2$ and the upper end $(\rho^2)=1$.  A similar computation can be done for the case of post-selection.  The result is that the change of the expectation value of the purity over a single time step is $3((\rho^2)-1)^2$, leading to:

\begin{align}
(\rho^2) = 1-\frac{1}{3t/N+2}\nonumber
\end{align}
Again, we can only trust this result for $t \ll N$.

\section{Appendix: Free fermion computations} \label{apfree}

The notation in this section is as in \cref{conserv}.  We compute the change in $S_{\mathrm{proxy}} (\rho)=(\log 2)(n - \Tr\,\mM^2)$, averaged over measurement outcomes, after measuring the occupation number of the first mode, i.e. the observable $i \gamma_1 \gamma_2$.  To facilitate the computation, let us first do a unitary rotation on modes $2$ through $n$ to diagonalize the lower right $n-1$ by $n-1$ block of $\mM$:

\begin{align}
\mM = \begin{pmatrix} 
    \eta_1 & w_2^* & \dots & w_n^* \\
    w_2 & \eta_2 & 0 & \dots  \\
    \vdots & 0 & \ddots &  \\
    w_n & \vdots &  & \eta_n
    \end{pmatrix}\nonumber
\end{align}
Note that $\Tr\,\mM^2 = \eta_1^2 + 2 \sum_{j=2}^n |w_j|^2 + \sum_{j=2}^n \eta_j^2$.

Let us first consider the case of the measurement outcome being that the first mode is empty (denoted with a plus subscript).  The probability of this outcome is $p_+ = \Tr\, \rho \, a_1 a_1^\dagger = (1+ \eta_1)/2$, so the normalized post-measurement state is $\rho'_+=(a_1 a_1^\dagger) \rho (a_1 a_1^\dagger)/p_+$.  Being a product of Gaussian states, $\rho'_+$ is Gaussian (see e.g. Ref. \cite{Bravyi}), and since pairing correlations still vanish, it is uniquely determined by $(\mM'_+)_{\mu \nu} = 2\, \Tr \, (\rho'_+ \, a_\mu a_\nu^\dagger) - \delta_{\mu\nu}$.  We immediately see that $(\mM'_+)_{1 1} = 1$, $(\mM'_+)_{\mu 1} = (\mM'_+)_{1 \mu} = 0$ for $\mu=2,\ldots, n$.  Using Wick's theorem, we have that for $\mu \geq 2$ and $\nu \geq 2$

\begin{align}
(\mM'_+)_{\mu\nu} &= \frac{2}{p_+} \Tr\, (\rho\, a_1 a_1^\dagger a_\mu a_\nu^\dagger ) -\delta_{\mu \nu} \\
 &= \frac{4}{1+\eta_1} \left[ \la a_1 a_1^\dagger \ra \la a_\mu a_\nu^\dagger \ra - \la a_1 a_\nu^\dagger \ra \la a_\mu a_1^\dagger \ra \right] - \delta_{\mu\nu} \\
 &= \frac{4}{1+\eta_1} \left[ \left(\frac{1+\eta_1}{2}\right) \left(\frac{\mM_{\mu\nu}+\delta_{\mu\nu}}{2}\right) - \frac{\mM_{1\nu}}{2} \frac{\mM_{\mu 1}}{2} \right] - \delta_{\mu\nu} \nonumber \\
 &= \delta_{\mu \nu} \eta_\mu - \frac{1}{1+\eta_1} w_\mu w_\nu^* \nonumber
\end{align}
A similar computation shows that the probability of the first mode being occupied is $p_-=(1-\eta_1)/2$, with the normalized post-measurement state having correlation matrix $(\mM'_-)_{1 1}= -1$, $(\mM'_-)_{\mu 1} = (\mM'_-)_{1  \mu}=0$ for $\mu=2,\ldots, n$, and $(\mM'_-)_{\mu \nu} = \delta_{\mu\nu} \eta_\mu + \frac{1}{1-\eta_1} w_\mu w_\nu^*$ for $\mu \geq 2$ and $\nu \geq 2$.

It is easiest to proceed using the Dirac notation, in which $\mM_+' = |1\ra\la1| + \sum_{j=2}^n \eta_j |j\ra\la j| - \frac{1}{1+\eta_1} |w\ra \la w|$, where $|w\ra = \sum_{j=2}^n w_j |j\ra$.  We then have

\begin{align}
(\mM_+')^2 = |1\ra\la 1| + \sum_{j=2}^n \eta_j^2 |j\ra \la j| - \frac{1}{1+\eta_1}\sum_{j=2}^n \left(\eta_j w_j |j\ra \la w| + \eta_j w_j^* |w\ra \la j| \right) + \frac{1}{(1+\eta_1)^2} |w|^2 |w\ra \la w| \nonumber
\end{align}
so that

\begin{align}
\Tr\, (\mM_+')^2 = 1+\sum_{j=2}^n \eta_j^2 - \frac{2}{1+\eta_1} \sum_{j=2}^n \eta_j |w|^2 + \frac{|w|^4}{(1+\eta_1)^2} \nonumber
\end{align}
Similarly
\begin{align}
\Tr\, (\mM_-')^2 = 1+\sum_{j=2}^n \eta_j^2 + \frac{2}{1-\eta_1} \sum_{j=2}^n \eta_j |w|^2 + \frac{|w|^4}{(1-\eta_1)^2} \nonumber
\end{align}
so that
\begin{align}
p_+ \Tr\, (\mM_+')^2 + p_- \Tr\, (\mM_-')^2 = 1+\sum_{j=2}^n \eta_j^2 + \frac{|w|^4}{1-\eta_1^2}. \nonumber
\end{align}
Thus
\begin{align}
\Delta S_{\mathrm{proxy}} &= (\log 2) \left(\Tr\,(\mM)^2 - p_+ \Tr\, (\mM_+')^2 - p_- \Tr\, (\mM_-')^2\right) \nonumber \\
&= - \frac{\log 2}{1-\eta_1^2} \left( 1- \eta_1^2 - |w|^2 \right)^2 \nonumber \\
&= - \frac{\log 2}{1-(\mM_{11})^2} (1-(\mM^2)_{11})^2 \leq 0 \nonumber
\end{align}

\bibliographystyle{apsrev4-1}
\nocite{apsrev41Control}
\bibliography{avent-ref}
\end{document}